\documentclass[12pt]{article}
\usepackage{amsmath}
\usepackage[latin1]{inputenc}
\usepackage{geometry}                % See geometry.pdf to learn the layout options. There are lots.
\usepackage{graphicx}
\usepackage{amssymb}
\usepackage{epstopdf}
\usepackage{amsthm}
\newtheorem{theorem}{Theorem}
\newtheorem{example}{Example}
\begin{document}

\title{Contextuality scenarios arising from networks of stochastic processes}
\date{}                                           % Activate to display a given date or no date

\author{
Rodrigo Iglesias\\
\small Instituto de Matem\'atica de Bah\'ia Blanca, CONICET\\ [-0.8ex]
\small Departamento de Matem\'atica, Universidad Nacional del Sur\\ [-0.8ex]
\small Bah\'ia Blanca, Argentina.
 \and
Fernando Tohm\'e\\
\small Instituto de Matem\'atica de Bah\'ia Blanca, CONICET\\ [-0.8ex]
\small Departamento de Econom\'ia, Universidad Nacional del Sur\\ [-0.8ex]
\small Bah\'ia Blanca, Argentina.
\and
Marcelo Auday\\
\small Instituto de Investigaciones Econ\'omicas y Sociales del Sur, CONICET\\ [-0.8ex]
\small Departamento de Humanidades, Universidad Nacional del Sur\\ [-0.8ex]
\small Bah\'ia Blanca, Argentina.
}

\maketitle
\begin{abstract}
An {\em empirical model} is a generalization of a {\em probability space}. It consists of a simplicial complex of subsets of a class $\mathcal{X}$ of random variables such that each simplex has an associated probability distribution. The ensuing marginalizations are coherent, in the sense that the distribution on a face of a simplex coincides with the marginal of the distribution over the entire simplex.

An empirical model is said {\em contextual} if its distributions cannot be obtained  marginalizing a joint distribution over $\mathcal{X}$. Contextual empirical models arise naturally in quantum theory, giving rise to some of its counter-intuitive statistical consequences.

In this paper we present a different and classical source of contextual empirical models: the interaction among many stochastic processes. We attach an empirical model to the ensuing network in which each node represents an open stochastic process with input and output random variables. The statistical behavior of the network in the long run makes the empirical model generically contextual and even strongly contextual.

\end{abstract}

\allowdisplaybreaks

 \section{Introduction}

 Let $\mathcal{X}= \{X_1,X_2,X_3\}$ be a set of random variables. A probability distribution of the joint random variable $(X_1,X_2,X_3)$ gives rise to a family of marginal distributions of $(X_1,X_2)$, $(X_2,X_3)$, $(X_3,X_1)$, $X_1$, $X_2$ and $X_3$. This family of distributions is said {\em consistent} in the sense that the distribution of $X_i$ is the marginal of $(X_j,X_i)$ as well as of $(X_i,X_k)$.

Not all consistent families of distributions on $(X_1,X_2)$, $(X_2,X_3)$ and $(X_3,X_1)$ arise as marginalizations of a joint distribution over $(X_1,X_2,X_3)$:
\begin{example}
 \label{example}
 Let $X_1, X_2, X_3$ be random variables with values in $\{0,1\}$ and let
 \begin{align*}
 P(X_1=0 \wedge X_2=0)=P(X_1=1 \wedge X_2=1) =  1/2 \\
 P(X_2=0 \wedge X_3=0)=P(X_2=1 \wedge X_3=1) =  1/2\\
 P(X_3=0 \wedge X_1=1)=P(X_3=1 \wedge X_1=0) = 1/2
 \end{align*}
This is a consistent family of distributions in the above sense but there is no probability assignment $P(X_1=0 \wedge X_2=0 \wedge X_3=0)$ of which they are its marginals.
\end{example}

This observation goes back to the seminal work of George Boole \cite{Boole}, who studied the conditions that a set of probabilities of logically related events must satisfy. Similar behaviors arise in other cases (most notably in quantum mechanics). An encompassing mathematical framework for their analysis is that of \textbf{empirical models} \cite{A-B}.

An empirical model consists of  a set $\mathcal{X}$ of random variables $X_1,...,X_n$, a family of subsets of $\mathcal{X}$ constituting an abstract simplicial complex, and a family of probability distributions, one for each simplex, satisfying a consistency condition. Namely, that if $\mathcal{U'} \subseteq \mathcal{U}$ are two simplices of the complex, the distribution attached to $\mathcal{U'}$ is identical to the marginal on $\mathcal{U'}$ of the distribution over $\mathcal{U}$.

Given a particular empirical model a fundamental question is whether there exists a joint distribution probability of $(X_1,...,X_n)$ such that all the distributions over subsets of variables are obtained as its marginalization. If such a joint distribution exists, the family of consistent distributions over the subsets is said {\em extendable} \cite{Vorobev}. If, on the contrary, the joint distribution does not exist, the empirical model is said to be {\em contextual}.

Vorobev \cite{Vorobev} gave a combinatorial characterization of those simplicial complexes
for which any family of consistent distributions is extendable. The failure in satisfying some of the assumptions that lead to this result allows the emergence of contextuality. This case, more interesting than the extendable setting, motivates the generalization of probability spaces to empirical models.

%To make this discussion more precise consider a probability space $(\Omega, B, P)$, where $\Omega$ is the sample space, $B$ is the $\sigma$-algebra of measurable sets and $P$ is the probability measure. An {\em extension} of $(\Omega, B, P)$ is a probability space $(\Omega', B', P')$ together with a surjective measurable map $\pi : \Omega' \rightarrow \Omega$ such that $P(E)=P'(\pi^{-1}E)$ for every $E \in B$. The distribution $P$ is a {\em marginal} of $P'$. Marginalization and extension are fundamental operations in probability theory. According to Tao \cite{Tao}, ``probability theory is only `allowed' to study concepts and perform operations which are preserved with respect to extension of the underlying sample space''. A probability notion should be invariant under extensions in the same way a notion in differential geometry is invariant under diffeomorphisms. Following this dogma, to go from probability spaces to empirical models is analogue to the transition from  $\mathbb{R}^n$ to differentiable manifolds.

The main motivation for considering this generalization of the notion of probability space comes from quantum theory. Many contextual empirical models have been shown to agree with the statistical predictions of quantum mechanics concerning certain experimental designs \cite{A}. According to quantum mechanics there are physically realizable families of consistent distributions of random variables which cannot be obtained as the marginals of the joint distribution of all the variables.

The formalism of quantum theory has also been used to describe probabilities in macroscopic natural phenomena beyond quantum mechanics. For an extensive developement in this direction see \cite{Kr}.

We analyze here a well motivated mathematical model, other than quantum theory, leading also to contextual empirical models. This alternative framework is classical, in the sense that any physical instantiation does not require quantum phenomena. 
It involves the interactions among many stochastic processes. In particular, our construction can be seen as a vast generalization of Example \ref{example}.

The basic unit in our framework is an {\em open stochastic process}. It models a device that receives the values of a set of input variables and --depending on the values of a set of internal values--  generates the values of a set of output variables according to a probability distribution encoded in a stochastic matrix. A pictorical representation of an open stochastic process is given in Figure \ref{1}.

\begin{figure}[ht]
\centering
\includegraphics[scale=0.6]{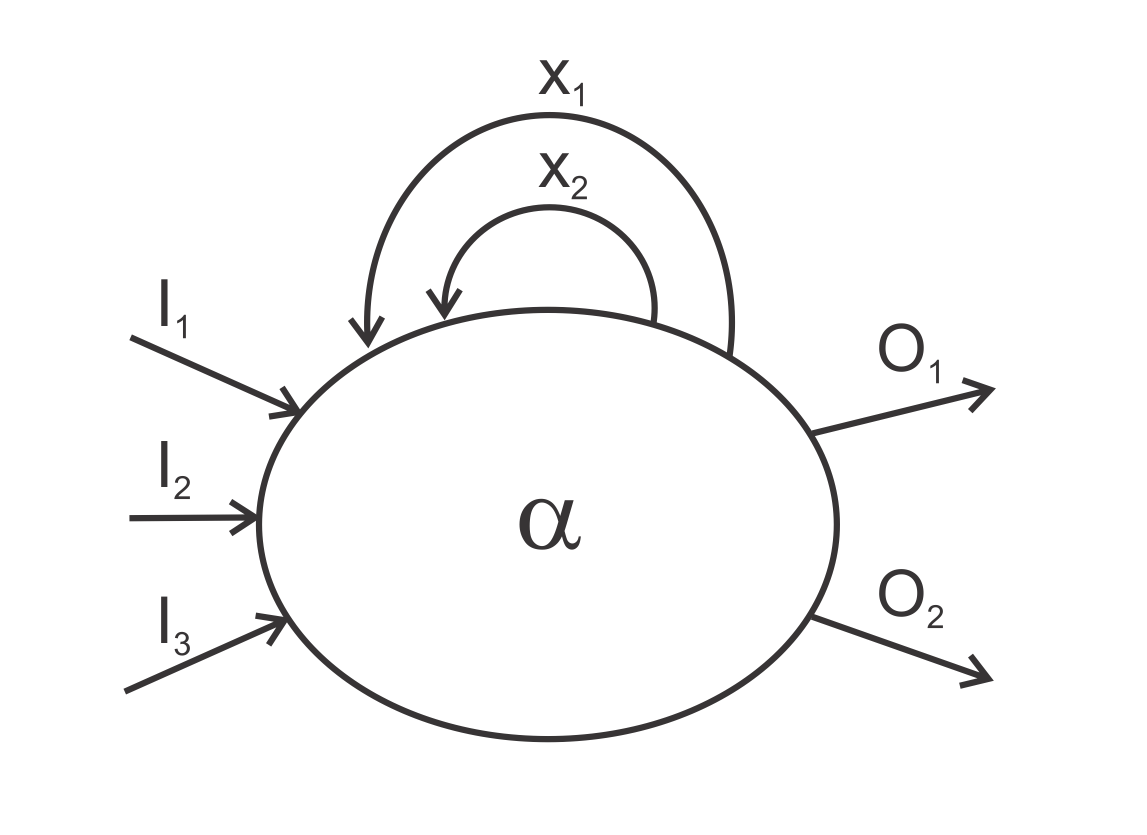}
\caption{ $\alpha$ represents an open stochastic process. Arrows $I_1,I_2,I_3$ represent the input variables, $X_1, X_2$ the internal variables and $O_1, O_2$ the output variables.}
\label{1}
\end{figure}

Open stochastic processes can be composed to build more complex stochastic processes connecting some of the output variables with some of the input variables of another process
%in the same way as the elementary gates of a circuit are composed to build a complex boolean %circuit
, as  illustrated in Figure \ref{7}.

\begin{figure}[ht]
\centering
\includegraphics[scale=0.6]{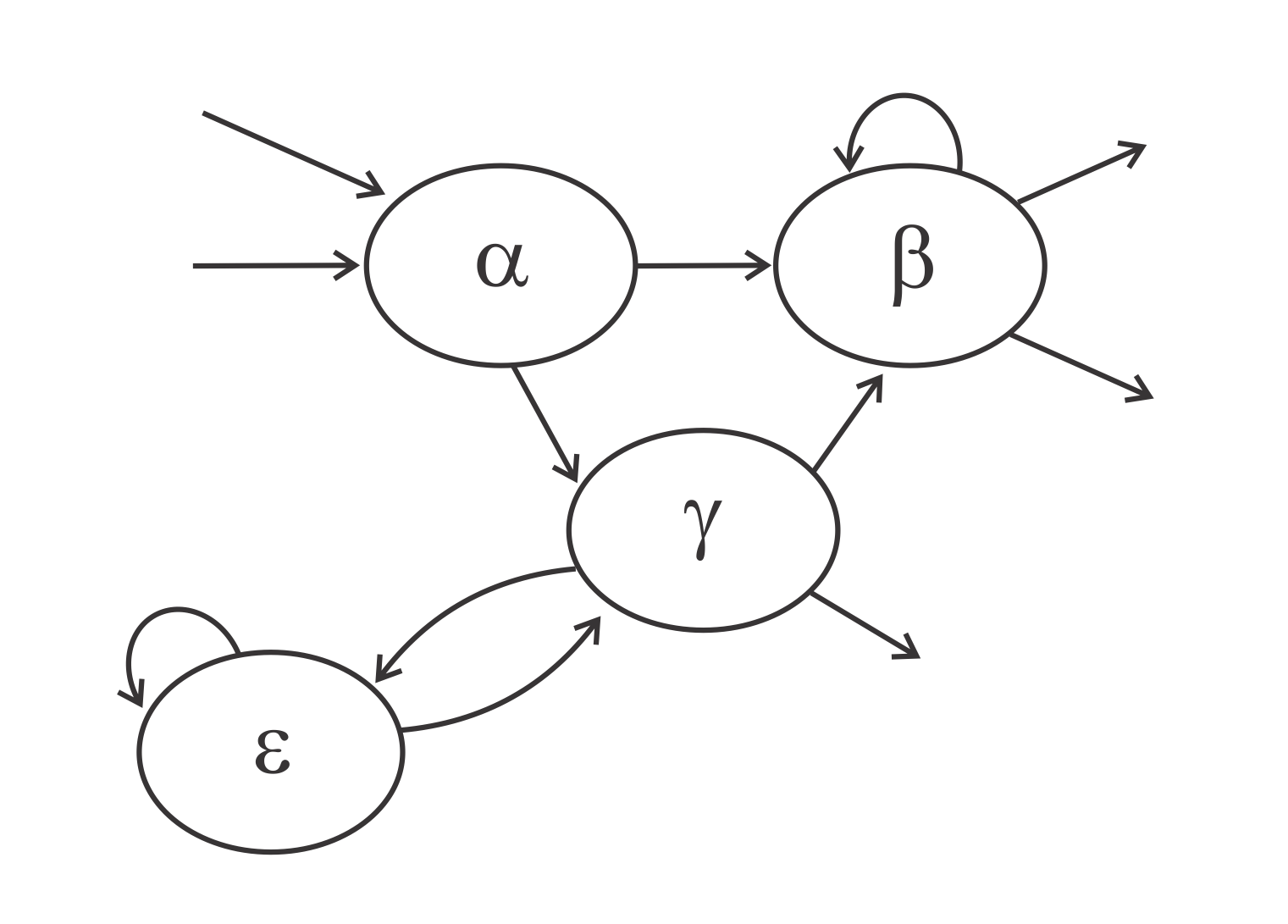}
\caption{A network of stochastic processes. Each arrow represents a variable. An output variable of one process can become an input variable of another one.}
\label{7}
\end{figure}

Under this representation our construction shows close resemblances to tensor networks and the operation of tensor contraction developed by the community of quantum computational complexity \cite{M-S} ,\cite{Biamonte}.
We can provide a graphical representation of the network of stochastic processes depicting each process as a node and each variable as an arrow.
The resulting network yields a single large process, as represented in Figure \ref{2}.

\begin{figure}[ht]
\centering
\includegraphics[scale=0.6]{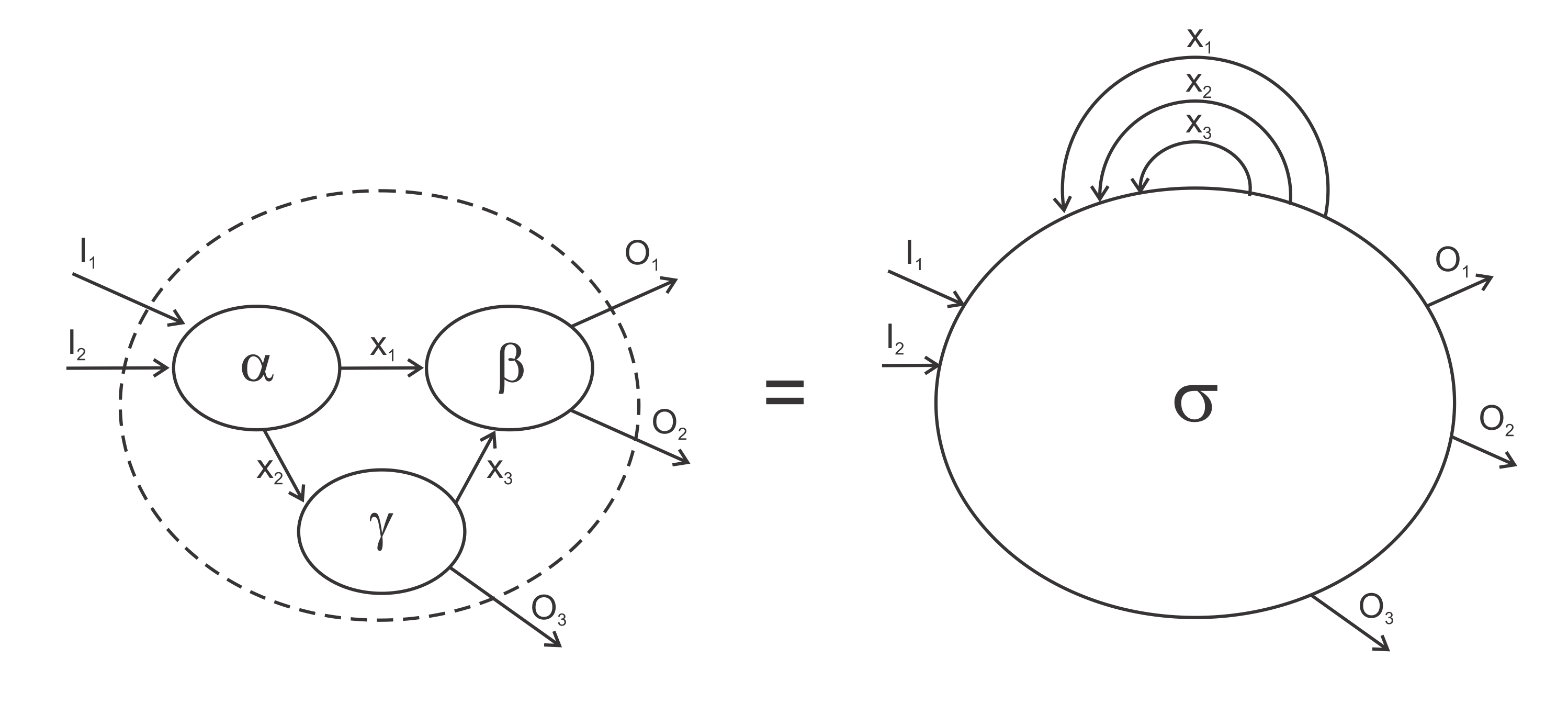}
\caption{A network composed by stochastic processes $\alpha, \beta, \gamma$ that yields a single open stochastic process $\sigma$. Any variable that is  simultaneously input of a node and output of another, becomes an internal variable of the resulting global process $\sigma$.}
\label{2}
\end{figure}

 We focus on the case in which this large process is closed, meaning that the resulting network has no input or output variables and thus all the variables are internal. In this sense, the large closed stochastic process can be seen as a large Markov chain.

We can analyze the statistical behavior of this Markov chain at the local level --i.e. at each node of the network--, namely the probabilistic features of the relation between the input and output variables of each individual open stochastic process. More specifically, we consider the probability $P^n_t(i,j)$ --as time goes to infinity-- of the event that at time $t$ the vector of input variables of node $n$ has value $i$ and at $t+1$ its vector of output variables gets value $j$. We obtain, in this way, a distribution of possible input-output values at each node. Our main result, Theorem \ref{Main result} indicates that this family of distributions constitutes an empirical model.

This paper is organized as follows. Section~\ref{empiricalsec} presents the basic elements in the description of empirical models. Section~\ref{netstoc0} describes the structure and behavior of networks of stochastic processes while Section~\ref{empstoc0} recasts this information in the framework of empirical models. Section~\ref{context} illustrates how a contextual empirical model can be defined upon a network of stochastic processes. 
Section~\ref{Comparison to quantum contextuality} shows that the contextuality arising from our model can be strictly stronger than that of quantum systems.
Section~\ref{conclusions} concludes the paper.

 \section{Empirical models}~\label{empiricalsec}

The notion of empirical model  was defined in \cite{A-B} as a formal framework in which the weird predictions of quantum mechanics can be made sense. A previous concept, useful for the definition of empirical models, is that of {\em measurement scenario} consisting of a finite set $\mathcal{X}$, a finite set $\mathcal{O}_X$ for each $X \in \mathcal{X}$, and a family $\mathcal{M}$ of subsets of $\mathcal{X}$ that covers $\mathcal{X}$ and such that if $\mathcal{U}' \subseteq \mathcal{U}$ and $\mathcal{U} \in \mathcal{M}$ then $\mathcal{U}' \in \mathcal{M}$.

The elements of $\mathcal{X}$ are called the {\em observables} or {\em variables}. When an observable $X$ is measured, an {\em outcome} from the set $\mathcal{O}_X$ is obtained. Each subset $\mathcal{U}$ in the family $\mathcal{M}$ is called a {\em context} and it represents a subset of observables that can be jointly measured. We say that $\mathcal{U}$ is a {\em maximal context} if it is not properly contained in another context.

Let $\mathcal{U}$ be a subset of $\mathcal{X}$. A section $s$ over $\mathcal{U}$ is an element of the cartesian product
$$
\mathcal{E}(\mathcal{U}) = \prod_{X \in \mathcal{U}} \mathcal{O}_X
$$
A {\em section} over $\mathcal{U}$ represents the outcome of the joint measurement of the variables in $\mathcal{U}$.
If $\mathcal{U}' \subseteq \mathcal{U}$ there exists a natural restriction map  $|_{\mathcal{U}'} : \mathcal{E}(\mathcal{U})\rightarrow \mathcal{E}(\mathcal{U}')$
$$
s \rightarrow s|_{\mathcal{U}'}
$$
Let us represent a {\em probability distribution over} $\mathcal{U}$ as a formal linear combination
$$
\pi= \sum_{s \in \mathcal{E}(U)} \pi_s \ s
$$
\noindent such that $\pi_s$ is a non-negative real number  for each $s$  and $\sum_{s \in \mathcal{E}(\mathcal{U})} \pi_s =1$. The number $\pi_s$ represents the probability of obtaining the outcome $s$ when the variables in $\mathcal{U}$ are measured. Let $\mathcal{D}(\mathcal{U})$ be the convex set of all probability distributions over $\mathcal{U}$.

The restriction map $|_{\mathcal{U}'} : \mathcal{E}(\mathcal{U})\rightarrow \mathcal{E}(\mathcal{U}')$ induces a map $|_{\mathcal{U}'} : \mathcal{D}(\mathcal{U})\rightarrow \mathcal{D}(\mathcal{U}')$ by linear extension. Namely, if
$
\pi= \sum_{s \in \mathcal{E}(\mathcal{\mathcal{U}})} \pi_s s
$
\noindent then
\begin{eqnarray}
\pi|_{\mathcal{U}'}\  = \sum_{s \in \mathcal{E}(\mathcal{U})} \pi_s \ s|_{\mathcal{U}'}
      \ =  \sum_{s' \in \mathcal{E}(\mathcal{U}')} \left(\sum_{ s|_{\mathcal{U}'}=s'} \pi_s \right) \ s'
\end{eqnarray}
In other words, if $\pi$ is a probability distribution over a set $\mathcal{U}$ of random variables and $\mathcal{U}'$ is a subset of $\mathcal{U}$, then $\pi|_{\mathcal{U}'}$ is the {\em marginal distribution} of $\pi$ when restricted to the variables in $\mathcal{U}'$.

Let $\mathcal{P}$ be a correspondence that assigns a probability distribution $\pi_\mathcal{U} \in  \mathcal{D}(\mathcal{U})$ to each context $\mathcal{U} \in \mathcal{M}$.
We say that the correspondence  $\mathcal{P}$ is a
 {\em no-signalling
 %\footnote{An expression stemming from quantum mechanics.}
  empirical model} for $\mathcal{M}$  if for every $\mathcal{U}' \subseteq \mathcal{U}$ with $\mathcal{U} \in \mathcal{M}$  we have the compatibility condition
$$
\pi_\mathcal{U} |_{\mathcal{U}'} = \pi_{\mathcal{U}'}
$$
For short, we refer to no-signalling empirical models simply as empirical models.
We say that a given empirical model is {\em non-contextual} if there exist a probability distribution $\pi_\mathcal{X}$ over $\mathcal{X}$ --the set of all the observables--  such that  for all $\mathcal{U} \in \mathcal{M}$
$$
\pi_\mathcal{U} = \pi_\mathcal{X} |_{\mathcal{U}}
$$
If such a global distribution fails to exist we say that the empirical model is {\em contextual}. As shown in Example \ref{example}, contextual empirical models can be easily constructed, by just postulating adequate distributions over subsets of variables. But random mechanisms generating this kind of behaviors are not pervasive. As noted above, quantum mechanics is a source of instances, but classical systems supporting them are much more rare. In the next section we present one, a network of stochastic processes.

\section{Networks of stochastic processes}~\label{netstoc0}

The representation of quantum processes by means of non-contextual empirical models relies on the properties of non-classical correlations among particles. A classical analogy involves the interaction among different random-generating subsystems. We will show how the system obtained through those interactions, represented as a network of stochastic processes, can be seen as an empirical model.
% and derive the conditions for its non-contextuality.

\subsection{Stochastic processes}
Intuitively, a process is a device or agent that receives as input the values of a set of variables and --depending on these values and its internal state--  generates the values of another set of output variables. If this process is not deterministic and follow a probability distribution, we say that it constitutes a stochastic process.

More precisely, a {\em stochastic process} consists of a set of input variables $I_1,...,I_n$, a set of internal variables $X_1,...,X_m$ and a set of output variables $O_1,...,O_r$, related by a stochastic matrix $\sigma$. We denote the entries of this matrix by
$$
\sigma
\left(
\begin{matrix}
I_1  & ... & I_n  &  X_1 & ... & X_m      \\
i'_1 & \dots & i'_n   &  x'_1 & \dots & x'_m
\end{matrix} \  \
\middle\vert \  \
\begin{matrix}
X_1 & \dots & X_m   & O_1 & \dots & O_r      \\
x_1 & \dots & x_m    & o_1 & \dots & o_r
\end{matrix}
\right)
$$
\noindent where this number represents the probability that, if at time $t$ we have $I_1=i_1,...,I_n=i_n$ and
$X_1=x'_1,...,X_m=x'_m$, then at time $t+1$  we have $X_1=x_1,...,X_m=x_m$ and $O_1=o_1,...,O_r=o_r$.
Note that the rows of the stochastic matrix $\sigma$ are indexed by tuples of input and internal variables, $(i'_1,...,i'_n,x'_1,...,x'_m)$, while the columns by tuples of internal and output variables, $(x_1,...,x_m,o_1,...,o_r)$.

This notation allows to distinguish easily the internal variables as those which appear simultaneously at both sides of the vertical line. The input variables, in turn, appear only at the left side of the vertical line but not at the right. The output variables, in turn, only appear at the right side of the vertical line but not at the left.

\subsection{Composition of stochastic processes}
We say that a stochastic process is {\em closed} if all its variables are internal. Otherwise, we say it is {\em open}.
The importance of open processes is that they can be combined to yield new processes by connecting output variables of some process to input variables of another. Let us see how this {\em composition} works in an example, which is illustrated in Figure \ref{3}.

\begin{figure}[ht]
\centering
\includegraphics[scale=0.6]{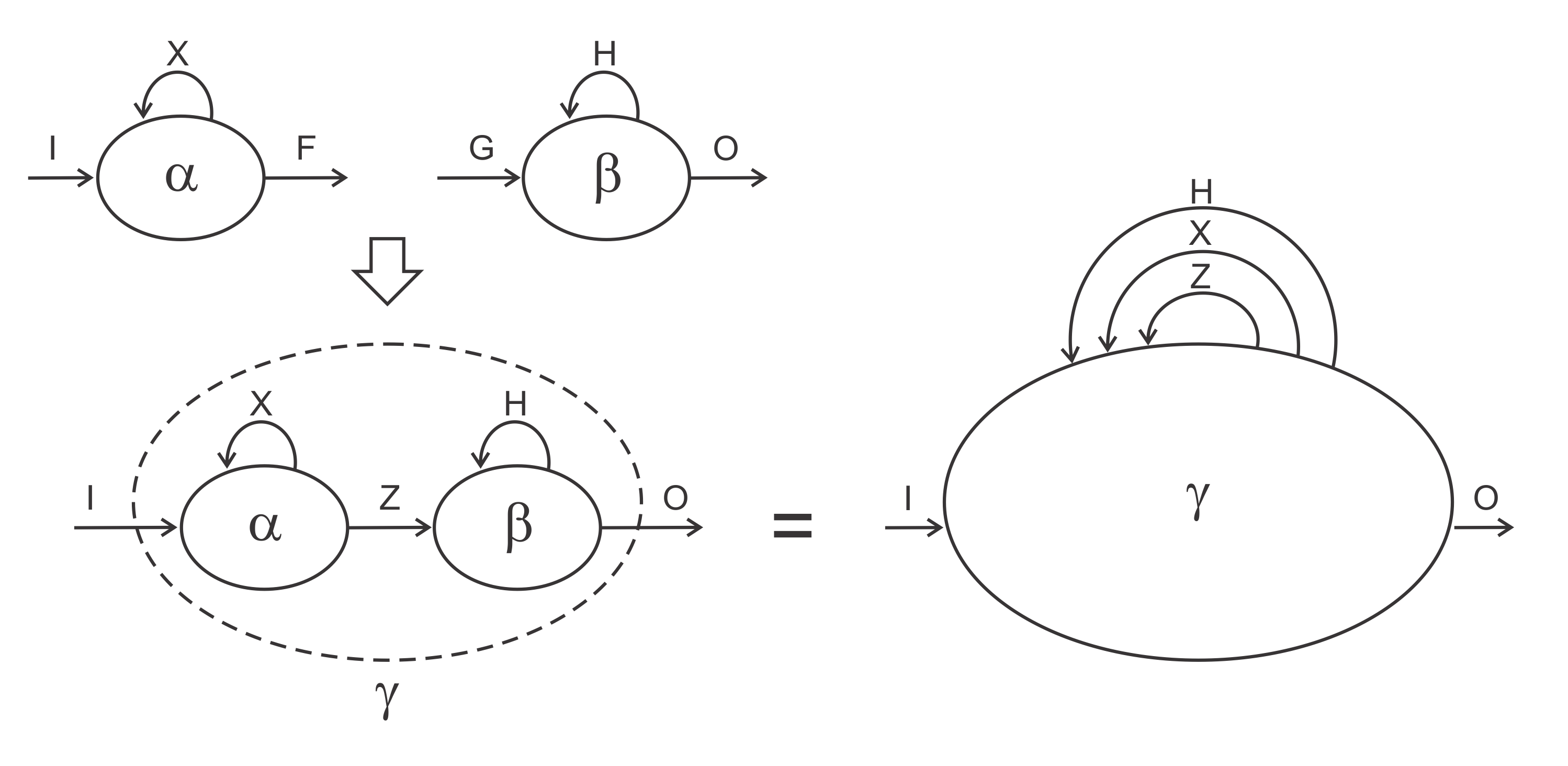}
\caption{Processes $\alpha$ and $\beta$ are composed by connecting the output variable $F$ of $\alpha$ with the input variable $G$ of $\beta$. The resulting process $\gamma$ has one input variable $I$, one output variable $O$ and three internal variables $X$, $Z$ and $H$.}
\label{3}
\end{figure}

Let $\alpha$ and $\beta$ be two open stochastic processes given by
$$
\alpha
\left(
\begin{matrix}
I &  X    \\
i' & x'
\end{matrix} \  \
\middle\vert \  \
\begin{matrix}
X &  F    \\
x & f
\end{matrix}
\right)
\ \ \ \ \mbox{and} \ \ \ \
\beta
\left(
\begin{matrix}
G &  H    \\
g' & h'
\end{matrix} \  \
\middle\vert \  \
\begin{matrix}
H &  O    \\
h & o
\end{matrix}
\right)
$$
\noindent respectively. If we connect the output variable $F$ of $\alpha$ to the input variable $G$ of $\beta$ we obtain the stochastic process $\gamma$ whose coefficients are given by the product of coefficients of the $\alpha$ and $\beta$ matrices:
$$
\gamma
\left(
\begin{matrix}
I &  X  & Z & H  \\
i' & x' & z' & h'
\end{matrix} \  \
\middle\vert \  \
\begin{matrix}
X &  Z   & H & O \\
x & z & h & o
\end{matrix}
\right)
\  = \
\alpha
\left(
\begin{matrix}
I &  X    \\
i' & x'
\end{matrix} \  \
\middle\vert \  \
\begin{matrix}
X &  Z    \\
x & z
\end{matrix}
\right)
\
\beta
\left(
\begin{matrix}
Z &  H    \\
z' & h'
\end{matrix} \  \
\middle\vert \  \
\begin{matrix}
H &  O    \\
h & o
\end{matrix}
\right)
$$

A key feature in this composition is that the output variable $F$ of $\alpha$ and the input variable $G$ of $\beta$ become jointly a single internal variable, $Z$, of $\gamma$.

\subsection{The network of processes}

Given two stochastic processes $\alpha$ and $\beta$ we say that $\alpha$ {\em provides} $\beta$ if there exists a variable $X$ which is at the same time an output variable of $\alpha$ and an input variable of $\beta$. A {\em reciprocity} is a pair $(\alpha, \beta)$ of stochastic processes such that $\alpha$ provides $\beta$ and $\beta$ provides $\alpha$ (see Figure \ref{5}).
In particular, $(\alpha, \alpha)$ is a reciprocity if and only if $\alpha$ has at least one internal variable.

\begin{figure}[ht]
\centering
\includegraphics[scale=0.6]{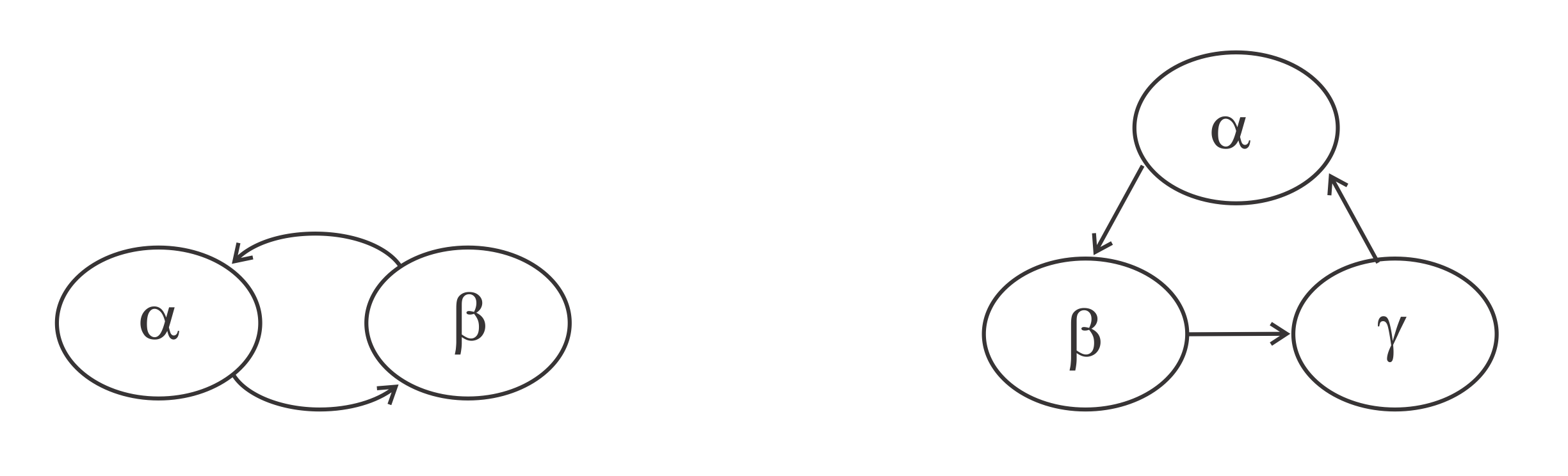}
\caption{Left: processes $\alpha$ and $\beta$ form a reciprocity. Right: a network without reciprocities.}
\label{5}
\end{figure}

We define a {\em network} $\mathcal{N}$ as a family $\{\sigma_i\}_{i=1,...,n}$ of stochastic processes yielding a stochastic process $\sigma$ obtained through the composition of members of $\{\sigma_i\}_{i=1,...,n}$. The elements of this family are called the {\em nodes} of the network and the variables of $\sigma$ are called the {\em arrows}. $\sigma$ is called the {\em global process}.
We say that $\mathcal{N}$ is a {\em network without reciprocities} if no pair $(\sigma_i,\sigma_j)$, for $i \neq j$, is a reciprocity (see Figure \ref{5}). We say that a network is {\em closed} if all the variables of the global process $\sigma$ are internal (see Figure \ref{4}).

\begin{figure}[ht]
\centering
\includegraphics[scale=0.6]{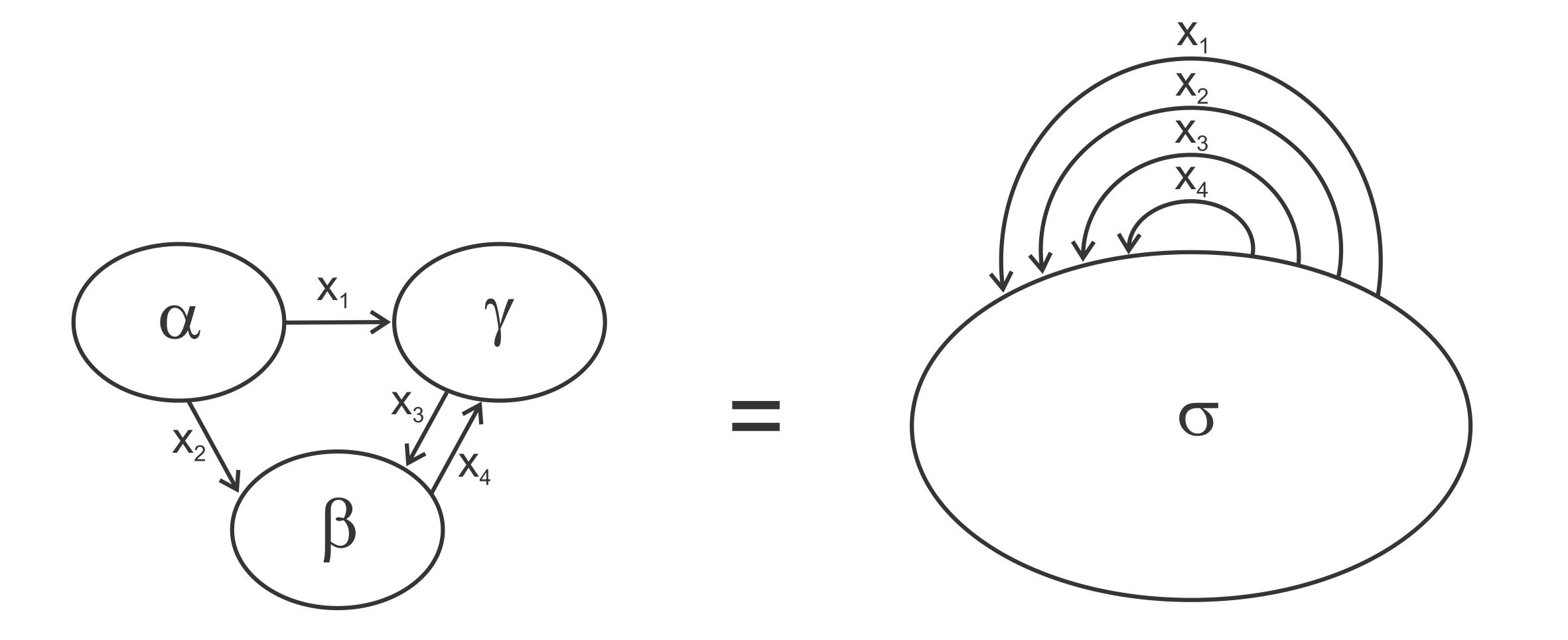}
\caption{A closed network. All the variables become internal variables of the global process $\sigma$.}
\label{4}
\end{figure}

It is useful to depict a network as an oriented  graph with multiple arrows where each node represents a process $\sigma_i$ and each arrow represents a variable. The input variables of the process $\sigma_i$ are represented by the incoming arrows of the node while its output variables are seen as arrows leaving the node. The input variables of the global process $\sigma$ are represented by arrows with no initial node. In turn, the output variables of $\sigma$ are arrows without terminal node. The internal variables of $\sigma$ are represented by those arrows connecting two nodes of  $\mathcal{N}$.

\subsection{Dynamics of closed stochastic processes}

Every closed stochastic process gives rise to a Markov chain. Let $\sigma$ be a closed stochastic process with (internal) variables $X_1,...,X_n$.  Let
$$
\pi_t
\left(
\begin{matrix}
  X_1 & \dots & X_n     \\
  x_1 & \dots & x_n
\end{matrix}
\right)
$$
\noindent be the probability that at time $t$ we have $X_1=x_1,...,X_n=x_n$. Then the evolution of this probability distribution is determined by
$$
\pi_{t+1}
\left(
\begin{matrix}
  X_1 & \dots & X_n     \\
  x_1 & \dots & x_n
\end{matrix}
\right)
= \sum_{x'_1...x'_n}
\sigma
\left(
\begin{matrix}
X_1 &  \dots & X_n    \\
x'_1 &  \dots & x'_n
\end{matrix} \  \
\middle\vert \  \
\begin{matrix}
X_1 &  \dots & X_n    \\
x_1 &  \dots & x_n
\end{matrix}
\right)
\
\pi_t
\left(
\begin{matrix}
  X_1 & \dots & X_n     \\
  x'_1 & \dots & x'_n
\end{matrix}
\right)
$$
For short, we write this equation as
$$
\pi_{t+1}
\left(
\begin{matrix}
  X    \\
  x
\end{matrix}
\right)
= \sum_{x'}
\sigma
\left(
\begin{matrix}
X    \\
x'
\end{matrix} \  \
\middle\vert \  \
\begin{matrix}
X    \\
x
\end{matrix}
\right)
\
\pi_t
\left(
\begin{matrix}
 X    \\
x'
\end{matrix}
\right)
$$

If we consider the stochastic matrix $\sigma$ as a linear transformation, we see that it transforms the convex set of all probability distributions into itself. Then, by the Brouwer's fixed point theorem\footnote{{\it ``For any continuous mapping $f: \Delta \rightarrow \Delta$, with $\Delta$ a compact and convex set, there exists $\bar{x} \in \Delta$ such that $f(\bar{x}) = \bar{x}$.''} In our case the mapping is given by $\sigma$, and $\Delta$ is the compact and convex set of probability distributions on the internal variables.}, there exists at least one probability distribution fixed by $\sigma$. Such a distribution is called a {\em stationary distribution}.\footnote{According to the Convergence Theorem of finite Markov chains (Theorem 4.9 of \cite{Levin}), the {\em irreducibility} and {\em aperiodicity} of  $\sigma$ are sufficient conditions for the uniqueness of a stationary distribution $\omega$, which is obtained as the limit
$$
\omega = \lim_{t\to\infty} \pi_t
$$
}

\section{The empirical model associated to a network}~\label{empstoc0}
From now on, let $\mathcal{N}$ be a closed network without reciprocities with a global process $\sigma$.
Given a node $\beta$ of the network $\mathcal{N}$, let $\mathcal{I}_{\beta}=\{I_1,...,I_n\}$ and $\mathcal{O}_{\beta}=\{O_1,...,O_m\}$  be respectively the set of input and output variables of $\beta$. Note that $\beta$ has no internal variables since $\mathcal{N}$ lacks reciprocities.  We will define a probability distribution  $\delta_{\beta}$ over $\mathcal{I}_{\beta} \cup \mathcal{O}_{\beta}$.
If the global process $\sigma$ has a unique stationary distribution, the coefficient
$$
\delta_{\beta}
\left(
\begin{matrix}
  I_1 & ... & I_n   &  O_1 & ... & O_m  \\
  i_1 & ... & i_n   &   o_1 & ... & o_m
\end{matrix}
\right)
$$
\noindent will be interpreted as the limit --when $t\to\infty$-- of the probability that at time $t$ we have $I_1=i_1,..., I_n=i_n$ and at time $t+1$ we have $O_1=o_1,..., O_m=o_m$ .

This distribution does not depend on the uniqueness of the stationary distribution $\omega$ of the global process $\sigma$.
Let $\omega|_{\mathcal{I}_{\beta}}$ be the marginal of $\omega$ over the subset of input variables $\mathcal{I}_{\beta}$. The distribution $\delta_{\beta}$ is defined by
$$
\delta_{\beta}
\left(
\begin{matrix}
  I_1 & \dots & O_m     \\
  i_1 & \dots & o_m
\end{matrix}
\right)
=
\beta
\left(
\begin{matrix}
I_1 &  \dots & I_n    \\
i_1 &  \dots & i_n
\end{matrix} \  \
\middle\vert \  \
\begin{matrix}
O_1 &  \dots & O_m    \\
o_1 &  \dots & o_m
\end{matrix}
\right)
\
\omega|_{\mathcal{I}_{\beta}}
\left(
\begin{matrix}
  I_1 & \dots & I_n     \\
  i_1 & \dots & i_n
\end{matrix}
\right)
$$

\begin{theorem}
\label{first theorem}
Let $\mathcal{N}$ be a closed network without reciprocities. Let $\omega$ be a stationary distribution of the entire process of  $\mathcal{N}$. Let $\beta$ be a node of  $\mathcal{N}$ and let $\mathcal{I}_{\beta}$, $\mathcal{O}_{\beta}$ and $\delta_{\beta}$ be defined as above.
Then
the marginal distributions of $\delta_{\beta}$ over $\mathcal{I}_{\beta}$ and $\mathcal{O}_{\beta}$ respectively are given by
$$
\delta_{\beta}|_{\mathcal{I}_{\beta}} = \omega|_{\mathcal{I}_{\beta}}  \ \ \ \ \ \ \ \ \ \ \ \
\delta_{\beta}|_{\mathcal{O}_{\beta}} = \omega|_{\mathcal{O}_{\beta}}
$$
\end{theorem}
\begin{proof}
Let $I$ be the joint of the variables $I_1,...,I_n$ and let $O$ be the joint of the variables $O_1,...,O_m$.
The first equality is obtained using the fact that node $\beta$ can be identified with its stochastic matrix:
\begin{align}
\delta_{\beta}|_{\mathcal{I}_{\beta}}
\left(
\begin{matrix}
  I   \\
  i
\end{matrix}
\right)
&=
\sum_{o}
\delta_{\beta}
\left(
\begin{matrix}
  I &O   \\
  i  & o
\end{matrix}
\right) \\
&=
\sum_{o}
\beta
\left(
\begin{matrix}
 I   \\
i
\end{matrix} \  \
\middle\vert \  \
\begin{matrix}
O    \\
o
\end{matrix}
\right)
\
\omega|_{\mathcal{I}_{\beta}}
\left(
\begin{matrix}
  I     \\
  i
  \end{matrix}
\right)  \\
&=
\omega|_{\mathcal{I}_{\beta}}
\left(
\begin{matrix}
  I     \\
  i
  \end{matrix}
\right)
\end{align}

For the second equality,
let $R$ be the joint of the variables of the global process $\sigma$ which are not contained in $\mathcal{I}_{\beta} \cup \mathcal{O}_{\beta}$. Let $\theta$ be the process obtained by restricting the network $\mathcal{N}$ to the nodes other than $\beta$ (see Figure \ref{6}).

\begin{figure}[ht]
\centering
\includegraphics[scale=0.6]{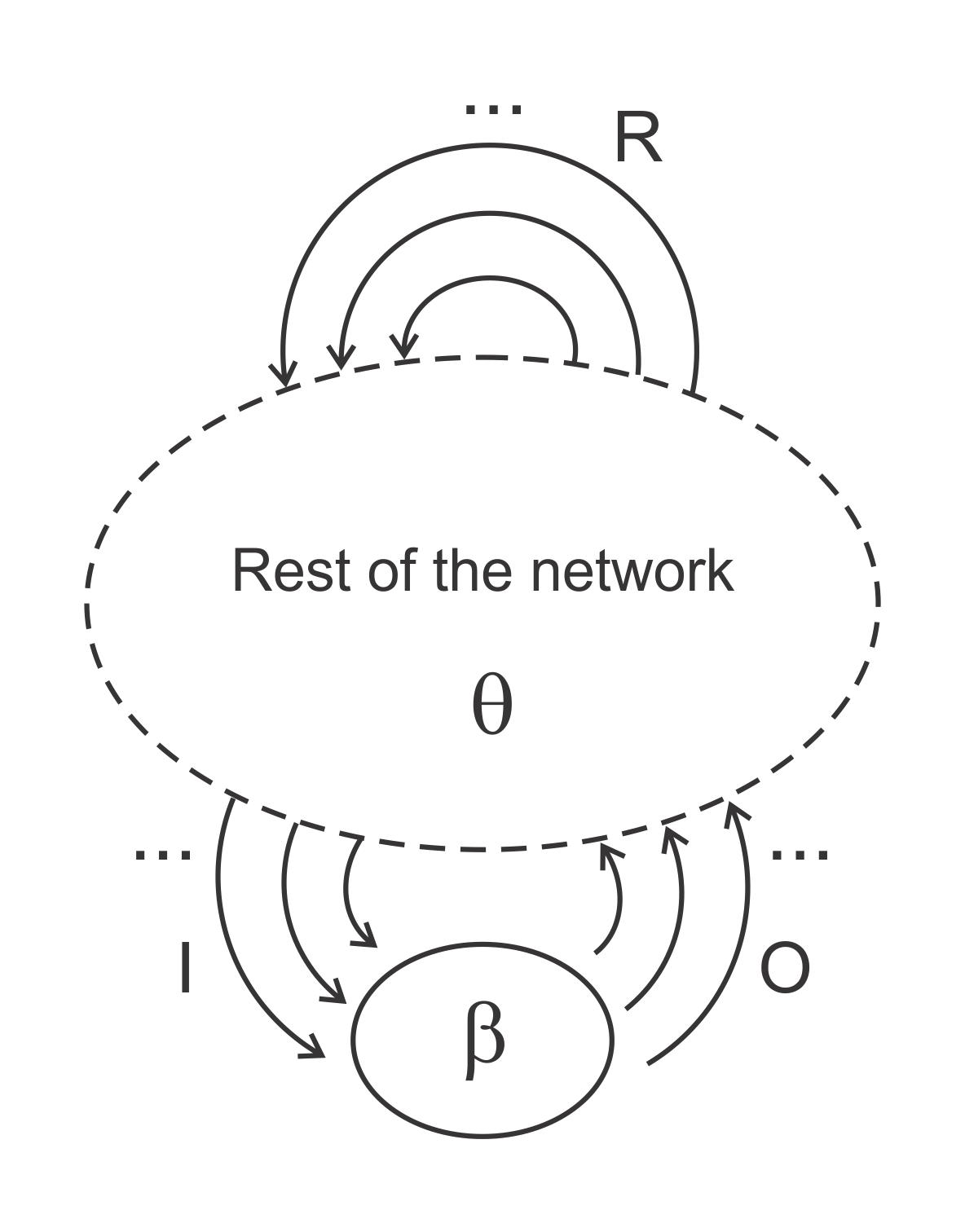}
\caption{$I$ is the joint of all the input variables of $\beta$, while $O$ is the joint of all its output variables. $R$ is the joint of the rest of the variables in network $\mathcal{N}$. Process $\theta$ is the restriction of $\mathcal{N}$ to the nodes that are not $\beta$.}
\label{6}
\end{figure}

This means that the entire process $\sigma$ is factorized as
$$
\sigma
\left(
\begin{matrix}
 I   & R & O\\
i'  & r' & o'
\end{matrix} \  \
\middle\vert \  \
\begin{matrix}
I & R & O    \\
i & r & o
\end{matrix}
\right) \
= \
\beta
\left(
\begin{matrix}
 I   \\
i'
\end{matrix} \  \
\middle\vert \  \
\begin{matrix}
O    \\
o
\end{matrix}
\right)
\
\
\theta
\left(
\begin{matrix}
 O & R  \\
o' & r'
\end{matrix} \  \
\middle\vert \  \
\begin{matrix}
R & I    \\
r & i
\end{matrix}
\right)
$$
 We will use the fact that $\omega$ is a stationary distribution of the process $\sigma$:
$$
\omega
\left(
\begin{matrix}
  I   & R & O  \\
  i   & r  & o
  \end{matrix}
\right)
\ =\
\sum_{i',r',o'}
\sigma
\left(
\begin{matrix}
 I & R & O \\
i' & r' & o'
\end{matrix} \  \
\middle\vert \  \
\begin{matrix}
I & R & O   \\
i & r & o
\end{matrix}
\right)
\
\omega
\left(
\begin{matrix}
  I   & R & O  \\
  i'   & r'   & o'
  \end{matrix}
\right)
$$
Then
\begin{align}
\delta_{\beta}|_{\mathcal{O}_{\beta}}
\left(
\begin{matrix}
  O  \\
  o
\end{matrix}
\right)
&=
\sum_{i'}
\delta_{\beta}
\left(
\begin{matrix}
  I &O   \\
  i'  & o
\end{matrix}
\right) \\
&=
\sum_{i'}
\beta
\left(
\begin{matrix}
 I   \\
i'
\end{matrix} \  \
\middle\vert \  \
\begin{matrix}
O    \\
o
\end{matrix}
\right)
\
\omega|_{\mathcal{I}}
\left(
\begin{matrix}
  I     \\
  i'
  \end{matrix}
\right)  \\
&=
\sum_{i'}
\beta
\left(
\begin{matrix}
 I   \\
i'
\end{matrix} \  \
\middle\vert \  \
\begin{matrix}
O    \\
o
\end{matrix}
\right)
\
\sum_{r',o'}
\omega
\left(
\begin{matrix}
  I   & R & O  \\
  i'   & r'   & o'
  \end{matrix}
\right)  \\
&=
\sum_{i'}
\beta
\left(
\begin{matrix}
 I   \\
i'
\end{matrix} \  \
\middle\vert \  \
\begin{matrix}
O    \\
o
\end{matrix}
\right)
\
\sum_{r',o'}
\sum_{i ,r}
\theta
\left(
\begin{matrix}
 O & R  \\
o' & r'
\end{matrix} \  \
\middle\vert \  \
\begin{matrix}
R & I    \\
r & i
\end{matrix}
\right)
\omega
\left(
\begin{matrix}
  I   & R & O  \\
  i'   & r'   & o'
  \end{matrix}
\right)  \\
&=
\sum_{i, r}
\sum_{i', r' ,o'}
\beta
\left(
\begin{matrix}
 I   \\
i'
\end{matrix} \  \
\middle\vert \  \
\begin{matrix}
O    \\
o
\end{matrix}
\right)
\
\theta
\left(
\begin{matrix}
 O & R  \\
o' & r'
\end{matrix} \  \
\middle\vert \  \
\begin{matrix}
R & I    \\
r & i
\end{matrix}
\right)
\omega
\left(
\begin{matrix}
  I   & R & O  \\
  i'   & r'   & o'
  \end{matrix}
\right)  \\
&=
\sum_{i,r}
\sum_{i',r',o'}
\sigma
\left(
\begin{matrix}
 I & R & O \\
i' & r' & o'
\end{matrix} \  \
\middle\vert \  \
\begin{matrix}
I & R & O   \\
i & r & o
\end{matrix}
\right)
\
\omega
\left(
\begin{matrix}
  I   & R & O  \\
  i'   & r'   & o'
  \end{matrix}
\right)
\\ &=
\sum_{i, r}
\omega
\left(
\begin{matrix}
  I   & R & O  \\
  i  & r  & o
  \end{matrix}
\right)  \\
&=
\omega|_{\mathcal{O}_{\beta}}
\left(
\begin{matrix}
   O  \\
  o
  \end{matrix}
\right)
\end{align}

\end{proof}

Given a network  without reciprocities $\mathcal{N}$, we attach to this network a measurement scenario as follows. The set of variables of this measurement scenario is the set of all the variables of the global  process $\sigma$. To each node $\beta$ in the network $\mathcal{N}$ corresponds a maximal context, $\mathcal{I}_{\beta} \cup \mathcal{O}_{\beta}$, which we denote by $\mathcal{U}_{\beta}$. They exhaust the set of maximal contexts of the measurement scenario.

\begin{theorem}
\label{Main result}
Let $\mathcal{N}$ be a closed network  without reciprocities with global process $\sigma$. Let $\omega$ be a stationary distribution of $\sigma$.
Let $\mathcal{P}$ be the correspondence that --for each node $\beta$ in the network $\mathcal{N}$-- assigns the distribution $\delta_{\beta}$ to the maximal context $\mathcal{U}_{\beta}$. Then $\mathcal{P}$ is a no-signalling empirical model.
\end{theorem}
\begin{proof}
It is sufficient to prove that given two nodes $\alpha$ and $\beta$ of the network we have
$$
\delta_{\alpha}|_{\mathcal{U}_{\alpha} \cap \mathcal{U}_{\beta}}
=
\delta_{\beta}|_{\mathcal{U}_{\alpha} \cap \mathcal{U}_{\beta}}
$$
Since the network has no reciprocities,
we can assume without loss of generality that all
the arrows in $\mathcal{U}_{\alpha} \cap \mathcal{U}_{\beta}$  go from $\alpha$ to $\beta$, that is, the variables in $\mathcal{U}_{\alpha} \cap \mathcal{U}_{\beta}$ are output variables of $\alpha$ and at the same time they are input variables of $\beta$. This means that
$$
\mathcal{U}_{\alpha} \cap \mathcal{U}_{\beta}  =  \mathcal{O}_{\alpha} \cap \mathcal{I}_{\beta}
$$
By Theorem \ref{first theorem} we have
\begin{align}
\delta_{\alpha}|_{\mathcal{U}_{\alpha} \cap \mathcal{U}_{\beta}}  &=
\delta_{\alpha}|_{\mathcal{O}_{\alpha} } |_{\mathcal{U}_{\alpha} \cap \mathcal{U}_{\beta}}
\\ &=
\omega|_{\mathcal{O}_{\alpha} }|_{\mathcal{U}_{\alpha} \cap \mathcal{U}_{\beta}}
\\ &=
\omega|_{\mathcal{U}_{\alpha} \cap \mathcal{U}_{\beta}}
\\ &=
\omega|_{\mathcal{I}_{\beta} }|_{\mathcal{U}_{\alpha} \cap \mathcal{U}_{\beta}}
\\ &=
\delta_{\beta}|_{\mathcal{I}_{\beta} } |_{\mathcal{U}_{\alpha} \cap \mathcal{U}_{\beta}}
\\ &=
\delta_{\beta}|_{\mathcal{U}_{\alpha} \cap \mathcal{U}_{\beta}}
\end{align}
\end{proof}

\section{Contextuality in networks of stochastic processes}~\label{context}

%Theorem \ref{Main result} shows that closed networks without reciprocities support an empirical model. If $\sigma$ is irreducible and aperiodic, $\omega$ would be the only stationary distribution. Even so, $\mathcal{P}$ can be contextual, since for each node $\beta$ of $\mathcal{N}$ the distribution $\delta_{\beta}$ is obtained by marginalizing $\omega$ only over over $\mathcal{I}_{\beta}$, but not over $\mathcal{O}_{\beta}$. More precisely, any $\delta_{\beta}$ differ from the full marginalization over $\mathcal{I}_{\beta}\cup \mathcal{O}_{\beta}$. The following example shows this:

%\begin{proposition}
%Let $\mathcal{N}$ be a closed network  without reciprocities with a global process $\sigma$ that is not irreducible or aperiodic. Then there exists a correspondence $\mathcal{P}$ which is a contextual no-signalling empirical model.
%\end{proposition}
%\begin{proof}
%Consider two different stationary distributions $\omega$ and $\omega^{\prime}$ that may exist since $\sigma$ is not both irreducible and aperiodic. Each one gives rise to a correspondence, $\mathcal{P}_{\omega}$ and $\mathcal{P}_{\omega^{\prime}}$ respectively. Now consider the marginal distributions $\delta_{\beta}$ associated to $\omega$. They cannot be seen as marginals of $\omega^{\prime}$. Thus, either $\mathcal{P}_{\omega^{\prime}}$ or $\mathcal{P}_{\omega}$ while being an empirical model, it is contextual.
%\end{proof}

%\subsection{Contextuality}

Theorem \ref{Main result} establishes that any stationary distribution in a closed network of stochastic processes without reciprocities gives rise to an empirical model. To determine whether such model is contextual or not requires further conditions. The main result in \cite{Vorobev} indicates that all empirical models over a given simplicial complex will be non-contextual if the latter is {\em regular}.\footnote{A simplicial complex $S$ is {\em regular} iff it belongs to the smallest class of simplicial complexes $\mathcal{S}$ such that: $(a)$ includes the class of all the proper subsets of vertices of cardinality $n$, for every possible $n$; and $(b)$ given $T \in \mathcal{S}$, $\hat{T} \in T$ and a set of vertices $A$ that does not belong to $T$, if we add all the subsets of $\hat{T} \cup A$ to $T$, we obtain another simplicial complex in $\mathcal{S}$ \cite{Vorobev}(def. 2.1).}  This implies, in our setting, that a necessary condition for the extendability of a given model is that the probability distributions over {\em all} the combinations of variables in the empirical model must be consistent. But in the network $\mathcal{N}$, supporting the empirical model, the requirement of consistency is imposed only on the internal distributions of the nodes and on those corresponding to connections among them. Furthermore, since the stochastic processes at the nodes are independent of each other, the distributions that matter for the definition of the empirical model have a natural local nature. This opens the possibility of defining a non-extendable empirical model with these features\footnote{A similar argument for the possible origin of contextuality in a combinatorial representation of quantum systems is presented in \cite{Fritz}.} Example \ref{example} illustrates that the focus on only some joint distributions on the set of variables (instead of over all possible combinations) can be such a source of contextuality.

We will in fact show that distributions like those of Example \ref{example} can be commonplace in our framework. This indicates that the empirical models supported by networks of stochastic processes are not extendable in general. That is, in a very natural sense we can say that contextuality arises as a generic property of our empirical models.

Let $\alpha$, $\beta$ and $\gamma$ be three open stochastic processes, each one with just one pair consisting of an input and an output variable, both assumed to be boolean.  That is, taking values in $\{T,F\}$. The corresponding matrices are:
$$
\alpha
\left(
\begin{matrix}
 X'  \\
x'
\end{matrix} \  \
\middle\vert \  \
\begin{matrix}
Y    \\
y
\end{matrix}
\right)
=
\left\{
	 \begin{array}{lll}
1 & \mbox{if}\ \  y=\neg x'  \\
0 & \mbox{otherwise}
\end{array}
\right.
$$
$$
\beta
\left(
\begin{matrix}
 Y'  \\
y'
\end{matrix} \  \
\middle\vert \  \
\begin{matrix}
Z    \\
z
\end{matrix}
\right)
=
\left\{
	 \begin{array}{lll}
1 & \mbox{if}\ \  z=y'  \\
0 & \mbox{otherwise}
\end{array}
\right.
$$
$$
\gamma
\left(
\begin{matrix}
 Z'  \\
z'
\end{matrix} \  \
\middle\vert \  \
\begin{matrix}
X    \\
x
\end{matrix}
\right)
=
\left\{
	 \begin{array}{lll}
1 & \mbox{if}\ \  x=z'  \\
0 & \mbox{otherwise}
\end{array}
\right.
$$

Now we compose these processes by connecting the output variable $Y$ of $\alpha$ with the input variable $Y'$ of $\beta$, as well as $Z$ of $\beta$ with $Z'$ of $\gamma$ and  $X$ of $\gamma$ to $X'$ of $\alpha$. The result is a closed network as shown in Figure \ref{8}.

\begin{figure}[ht]
\centering
\includegraphics[scale=0.6]{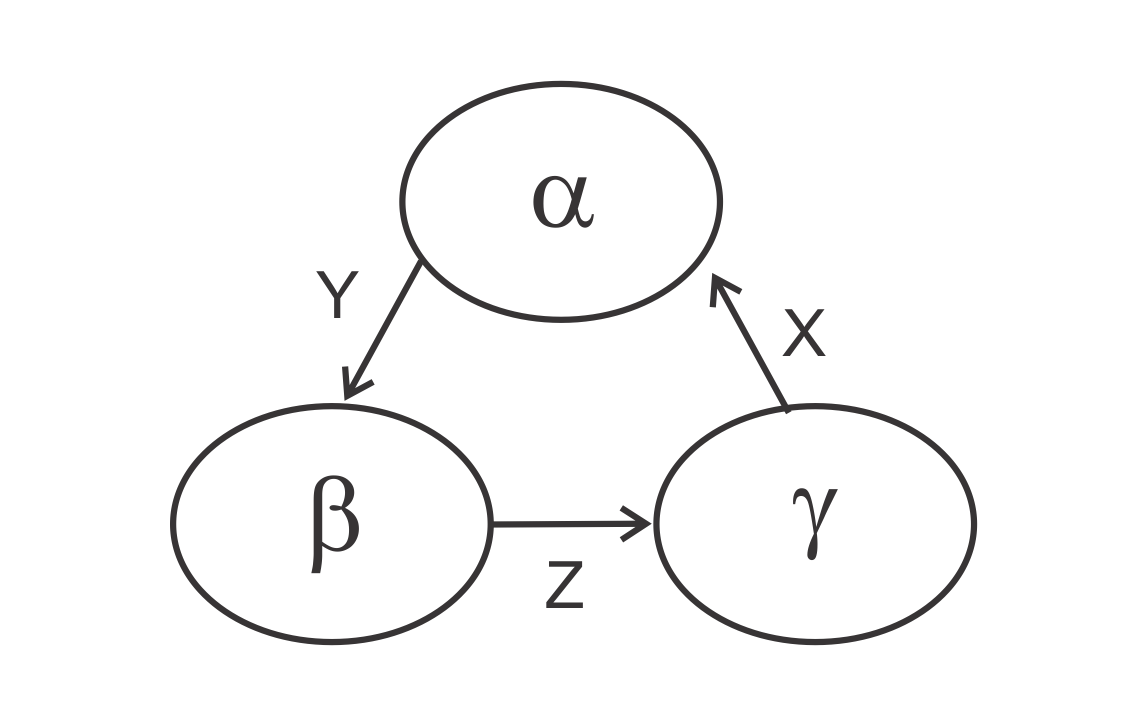}
\caption{Graphical representation of the example.}
\label{8}
\end{figure}

The global process $\sigma$ attached to this network is given by
$$
\sigma
\left(
\begin{matrix}
 X & Y & Z \\
x' & y' & z'
\end{matrix} \  \
\middle\vert \  \
\begin{matrix}
X & Y & Z   \\
x &y & z
\end{matrix}
\right)
=
\alpha
\left(
\begin{matrix}
 X  \\
x'
\end{matrix} \  \
\middle\vert \  \
\begin{matrix}
Y    \\
y
\end{matrix}
\right)
\ \
\beta
\left(
\begin{matrix}
 Y  \\
y'
\end{matrix} \  \
\middle\vert \  \
\begin{matrix}
Z    \\
z
\end{matrix}
\right)
\ \
\gamma
\left(
\begin{matrix}
 Z  \\
z'
\end{matrix} \  \
\middle\vert \  \
\begin{matrix}
X    \\
x
\end{matrix}
\right)
$$
that is,
$$
\sigma
\left(
\begin{matrix}
 X & Y & Z \\
x' & y' & z'
\end{matrix} \  \
\middle\vert \  \
\begin{matrix}
X & Y & Z   \\
x &y & z
\end{matrix}
\right)
=
\left\{
	 \begin{array}{lll}
1 & \mbox{if}\ \  (y= \neg x') \wedge (z=y') \wedge (x=z' ) \\
0 & \mbox{otherwise}
\end{array}
\right.
$$
The state space of the network is given by the eight possible values of the joint variable $(X,Y,Z)$. The dynamics determined by $\sigma$ can follow two cycles, one of which goes through the states $(x,y,z)$ such that $(x=y) \wedge (z=\neg x)$ and the other through those such that $(x=z) \wedge (y=\neg x)$.

The process $\sigma$ yields an infinite family of stationary distributions. Let us focus on the following one:
$$
\omega
\left(
\begin{matrix}
X & Y & Z   \\
x &y & z
\end{matrix}
\right)
=
\left\{
	 \begin{array}{lll}
0 & \mbox{if}\ \  (x=y) \wedge (z=\neg x) \\
1/6 & \mbox{otherwise}
\end{array}
\right.
$$
We can easily check that $\omega$ is indeed stationary.\footnote{That is, $\omega$ is a fixed point of the mapping given by
$$
\sum_{x',y',z'}
\sigma
\left(
\begin{matrix}
 X & Y & Z \\
x' & y' & z'
\end{matrix} \  \
\middle\vert \  \
\begin{matrix}
X & Y & Z   \\
x & y & z
\end{matrix}
\right)
\
\omega
\left(
\begin{matrix}
  X   & Y & Z  \\
  x'   & y'   & z'
  \end{matrix}
\right)
$$
}
Note that the marginals $\omega|_{X}$, $\omega|_{Y}$ and $\omega|_{Z}$ are uniform distributions, that is
$$
\omega|_{X}
\left(
\begin{matrix}
   X  \\
  x
  \end{matrix}
\right)
= 1/2 \ \ \mbox{for all} \ \ x
$$
$$
\omega|_{Y}
\left(
\begin{matrix}
   Y  \\
  y
  \end{matrix}
\right)
= 1/2 \ \ \mbox{for all} \ \ y
$$
$$
\omega|_{Z}
\left(
\begin{matrix}
   Z  \\
  z
  \end{matrix}
\right)
= 1/2 \ \ \mbox{for all} \ \ z
$$
Then, according to the construction leading to Theorem \ref{Main result}, the empirical model attached to the network is given by:\footnote{Note that, in general, the distributions $\delta_{\alpha}$,$\delta_{\beta}$ and $\delta_{\gamma}$ do not coincide with $\omega|_{X,Y}$, $\omega|_{Y,Z}$ and $\omega|_{Z,X}$ respectively. In this case we have, for instance that (analogously for $(Y,Z)$ and $(Z,X)$)
$$
\omega|_{X,Y}
\left(
\begin{matrix}
X &Y   \\
x  & y
\end{matrix}
\right)
=
\left\{
\begin{array}{lll}
1/6 & \mbox{if}\ \  x=y \\
1/3 & \mbox{otherwise}
\end{array}
\right.
$$
}

$$
\delta_{\alpha}
\left(
\begin{matrix}
  X &Y   \\
  x  & y
\end{matrix}
\right)
= \ \
\alpha
\left(
\begin{matrix}
 X  \\
x
\end{matrix} \  \
\middle\vert \  \
\begin{matrix}
Y    \\
y
\end{matrix}
\right)
\ \
\omega|_{X}
\left(
\begin{matrix}
   X  \\
  x
  \end{matrix}
\right)
=
\left\{
	 \begin{array}{lll}
1/2 & \mbox{if}\ \  x=\neg y \\
0 & \mbox{otherwise}
\end{array}
\right.
$$

$$
\delta_{\beta}
\left(
\begin{matrix}
  Y &Z   \\
  y  & z
\end{matrix}
\right)
= \ \
\beta
\left(
\begin{matrix}
 Y  \\
y
\end{matrix} \  \
\middle\vert \  \
\begin{matrix}
Z    \\
z
\end{matrix}
\right)
\ \
\omega|_{Y}
\left(
\begin{matrix}
   Y  \\
  y
  \end{matrix}
\right)
=
\left\{
	 \begin{array}{lll}
1/2 & \mbox{if}\ \  y=z \\
0 & \mbox{otherwise}
\end{array}
\right.
$$

$$
\delta_{\gamma}
\left(
\begin{matrix}
  Z &X   \\
  z  & x
\end{matrix}
\right)
= \ \
\gamma
\left(
\begin{matrix}
 Z  \\
z
\end{matrix} \  \
\middle\vert \  \
\begin{matrix}
X    \\
x
\end{matrix}
\right)
\ \
\omega|_{Z}
\left(
\begin{matrix}
   Z  \\
  z
  \end{matrix}
\right)
=
\left\{
	 \begin{array}{lll}
1/2 & \mbox{if}\ \  z=x \\
0 & \mbox{otherwise}
\end{array}
\right.
$$
These distributions over $(X,Y)$, $(Y,Z)$ and $(Z,X)$ are the same as those of Example \ref{example}.

\section{Comparison to quantum contextuality }
\label{Comparison to quantum contextuality}

The previous example suggests that networks of processes can support empirical models exhibiting a strong form of contextuality. As a way of gauging the strength of contextuality, in this section we consider the familiar Clauser-Horne-Shimony-Holt (CHSH) quantum scenario
 \cite{CHSH}.

The CHSH scenario consists of four observables $A_1$, $A_2$, $B_1$ and $B_2$ together with four contexts $(A_1,B_1)$, $(A_2,B_1)$, $(A_1,B_2)$ and $(A_2,B_2)$. The set of outcomes of each variable is $\{0,1\}$. This empirical model can be defined in our framework as follows. We attach a stochastic process $\alpha_{ij}$ to each context $(A_i,B_j)$. The set of outcomes of each variable is $\{0,1\}$. The matrices defining these stochastic processes are given by

$$
\alpha_{11}
\left(
\begin{matrix}
 A'_1  \\
a'_1
\end{matrix} \  \
\middle\vert \  \
\begin{matrix}
B_1    \\
b_1
\end{matrix}
\right)
=
\left\{
	 \begin{array}{lll}
1 & \mbox{if}\ \  a'_1=\neg b_1  \\
0 & \mbox{otherwise}
\end{array}
\right.
$$
$$
\alpha_{21}
\left(
\begin{matrix}
 B'_1  \\
b'_1
\end{matrix} \  \
\middle\vert \  \
\begin{matrix}
A_2    \\
a_2
\end{matrix}
\right)
=
\left\{
	 \begin{array}{lll}
1 & \mbox{if}\ \  b_1= a'_2  \\
0 & \mbox{otherwise}
\end{array}
\right.
$$
$$
\alpha_{22}
\left(
\begin{matrix}
 A'_2  \\
a'_2
\end{matrix} \  \
\middle\vert \  \
\begin{matrix}
B_2    \\
b_2
\end{matrix}
\right)
=
\left\{
	 \begin{array}{lll}
1 & \mbox{if}\ \  a'_2= b_2  \\
0 & \mbox{otherwise}
\end{array}
\right.
$$
$$
\alpha_{12}
\left(
\begin{matrix}
 B'_2  \\
b'_2
\end{matrix} \  \
\middle\vert \  \
\begin{matrix}
A_1    \\
a_1
\end{matrix}
\right)
=
\left\{
	 \begin{array}{lll}
1 & \mbox{if}\ \  b'_2=\neg a_1  \\
0 & \mbox{otherwise}
\end{array}
\right.
$$

As in the previous section, we compose these processes to build a closed network. Namely, we connect the variable $A'_i$ with $A_i$ and the variable $B'_i$ with $B_i$, for $i=1,2$.
The resulting global process $\sigma$ is given by

$$
\sigma
\left(
\begin{matrix}
 A_1 & B_1 & A_2 & B_2 \\
a'_1 & b'_1 & a'_2 & b'_2
\end{matrix} \  \
\middle\vert \  \
\begin{matrix}
A_1 & B_1 & A_2 & B_2 \\
a_1 & b_1 & a_2 & b_2
\end{matrix}
\right)
=
$$
$$
\alpha_{11}
\left(
\begin{matrix}
 A_1  \\
a'_1
\end{matrix} \  \
\middle\vert \  \
\begin{matrix}
B_1    \\
b_1
\end{matrix}
\right)
\ \
\alpha_{21}
\left(
\begin{matrix}
 B_1  \\
b'_1
\end{matrix} \  \
\middle\vert \  \
\begin{matrix}
A_2    \\
a_2
\end{matrix}
\right)
\ \
\alpha_{22}
\left(
\begin{matrix}
 A_2  \\
a'_2
\end{matrix} \  \
\middle\vert \  \
\begin{matrix}
B_2    \\
b_2
\end{matrix}
\right)
\alpha_{12}
\left(
\begin{matrix}
 B_2  \\
b'_2
\end{matrix} \  \
\middle\vert \  \
\begin{matrix}
A_1    \\
a_1
\end{matrix}
\right) =
$$

$$
=
\left\{
	 \begin{array}{lll}
1 & \mbox{if}\ \  (a'_1= \neg b_1) \wedge (b'_1=a_2) \wedge (a'_2=b_2 )  \wedge (b'_2=a_1 ) \\
0 & \mbox{otherwise}
\end{array}
\right.
$$

Since this is a permutation matrix, the uniform distribution on the state space of the joint variable $(A_1, B_1, A_2, B_2)$ is stationary. We can thus define

$$
\omega
\left(
\begin{matrix}
A_1 & B_1 & A_2 & B_2   \\
a_1 & b_1 & a_2 & b_2 
\end{matrix}
\right)
= 1/16
$$
for all $a_1 , b_1 , a_2 , b_2 $.
The marginals $\omega|_{A_1}$, $\omega|_{B_1}$, $\omega|_{A_2}$ and $\omega|_{B_2}$ are uniform distributions. Then the empirical model attached to the network is given by

$$
\delta_{\alpha_{11}}
\left(
\begin{matrix}
  A_1 & B_1   \\
  a_1  & b_1
\end{matrix}
\right)
= \ \
\alpha_{11}
\left(
\begin{matrix}
 A_1  \\
a_1
\end{matrix} \  \
\middle\vert \  \
\begin{matrix}
B_1    \\
b_1
\end{matrix}
\right)
\ \
\omega|_{A_1}
\left(
\begin{matrix}
   A_1  \\
  a_1
  \end{matrix}
\right)
=
\left\{
	 \begin{array}{lll}
1/2 & \mbox{if}\ \  a_1=\neg b_1 \\
0 & \mbox{otherwise}
\end{array}
\right.
$$
$$
\delta_{\alpha_{21}}
\left(
\begin{matrix}
  B_1 & A_2   \\
  b_1  & a_2
\end{matrix}
\right)
= \ \
\alpha_{21}
\left(
\begin{matrix}
 B_1  \\
b_1
\end{matrix} \  \
\middle\vert \  \
\begin{matrix}
A_2    \\
a_2
\end{matrix}
\right)
\ \
\omega|_{B_1}
\left(
\begin{matrix}
   B_1  \\
  b_1
  \end{matrix}
\right)
=
\left\{
	 \begin{array}{lll}
1/2 & \mbox{if}\ \  b_1= a_2 \\
0 & \mbox{otherwise}
\end{array}
\right.
$$
$$
\delta_{\alpha_{22}}
\left(
\begin{matrix}
  A_2 & B_2   \\
  a_2  & b_2
\end{matrix}
\right)
= \ \
\alpha_{22}
\left(
\begin{matrix}
 A_2  \\
a_2
\end{matrix} \  \
\middle\vert \  \
\begin{matrix}
B_2    \\
b_2
\end{matrix}
\right)
\ \
\omega|_{A_2}
\left(
\begin{matrix}
   A_2  \\
  a_2
  \end{matrix}
\right)
=
\left\{
	 \begin{array}{lll}
1/2 & \mbox{if}\ \  a_2= b_2 \\
0 & \mbox{otherwise}
\end{array}
\right.
$$
$$
\delta_{\alpha_{12}}
\left(
\begin{matrix}
  B_2 & A_1   \\
  b_2  & a_1
\end{matrix}
\right)
= \ \
\alpha_{12}
\left(
\begin{matrix}
 B_2  \\
b_2
\end{matrix} \  \
\middle\vert \  \
\begin{matrix}
A_1    \\
a_1
\end{matrix}
\right)
\ \
\omega|_{B_2}
\left(
\begin{matrix}
   B_2  \\
  b_2
  \end{matrix}
\right)
=
\left\{
	 \begin{array}{lll}
1/2 & \mbox{if}\ \  b_2= a_1 \\
0 & \mbox{otherwise}
\end{array}
\right.
$$
We can see that this coincides with the definition of Popescu-Rohrlich (PR) box \cite{PR}:
\begin{align*}
 P(A_1=0 \wedge B_1=1)=P(A_1=1 \wedge B_1=0) =  1/2 \\
 P(A_1=0 \wedge B_2=0)=P(A_1=1 \wedge B_2=1) =  1/2\\
 P(A_2=0 \wedge B_1=0)=P(A_2=1 \wedge B_1=1) = 1/2\\
 P(A_2=0 \wedge B_2=0)=P(A_2=1 \wedge B_2=1) = 1/2
 \end{align*}

It is known that the PR box constitutes a no-signalling model that achieves super-quantum correlations by violating Tsirelson's \cite{bound} quantum contextuality bound  of $2\sqrt{2}$ for the CHSH value \cite{Marco} . This means that our construction is not restricted to quantum contextuality, but it corresponds to that of generalized probabilistic theories (see for example \cite{Cabello}).

\section{Conclusions}~\label{conclusions}

While the result presented in the previous sections is just an instance, we can hint at its generality. Networks are combinations of arbitrary stochastic processes and the variables of non directly connected nodes can be distributed independently. 

The previously known sources of contextuality were of quantum nature. In this paper we introduced a classical model, namely that of a network of stochastic processes. While the ensuing structure shares with quantum systems the property of contextuality it supports a stronger version of this property.

\end{document}